%% file: root.tex
\documentclass[letterpaper, 10 pt, conference]{ieeeconf}  % Comment this line out if you need a4paper

\IEEEoverridecommandlockouts                              % This command is only needed if 
\usepackage{amsmath,amssymb,amsfonts} % amsthm interferes with proof environment from ieeeconf.cls.
\usepackage{graphicx}
\usepackage{pgfplots}
\usepackage{booktabs}
\usepackage{textcomp}
\usepackage{xcolor}
\usepackage{siunitx}
\usepackage[inkscapeformat=png]{svg}
\usepackage{subcaption}
\usepackage[acronym]{glossaries}
\usepackage[noadjust]{cite}

\newacronym{cbf}{CBF}{Control Barrier Function}
\newacronym{cps}{CPS}{cyber-physical systems}

\title{\LARGE \bf Peak Bounds for the Estimation Error under Sensor Attacks}

\author{Axel Stafström$^{1}$ \and Daniel Arnström$^{1,2}$ \and Adam Miksits$^{2,3}$ \and David Umsonst$^{2}$
\thanks{This work was partially supported by the Wallenberg AI, Autonomous Systems and Software Program (WASP) funded by the Knut and Alice Wallenberg Foundation.}% <-this % stops a space
\thanks{$^{1}$Uppsala University, Sweden, 
        {\tt\small \{daniel.arnstrom@it, axel.stafstrom.9670@student\}.uu.se}}%
\thanks{$^{2}$Ericsson Research, Sweden, 
        {\tt\small \{adam.miksits, david.umsonst,daniel.arnstrom\}@ericsson.com}}%
\thanks{$^{3}$KTH Royal Institute of Technology, Sweden
        {\tt\small amiksits@kth.se}}%
}

\newtheorem{theorem}{Theorem}
\newtheorem{corollary}{Corollary}
\newtheorem{proposition}{Proposition}
\newtheorem{lemma}{Lemma}
\newtheorem{assumption}{Assumption}

\newtheorem{definition}{Definition}

\newtheorem{remark}{Remark}
\newtheorem{problem}{Problem}

\newcommand{\abs}[1]{\left\lvert#1\right\rvert} 
\newcommand{\norm}[1]{\left\lVert#1\right\rVert}

\newcommand{\R}{\mathbb{R}}

\let\epsilon\varepsilon

\begin{document}

\maketitle
\thispagestyle{empty}
\pagestyle{empty}

%%%%%%%%%%%%%%%%%%%%%%%%%%%%%%%%%%%%%%%%%%%%%%%%%%%%%%%%%%%%%%%%%%%%%%%%%%%%%%%%
\begin{abstract}

This paper investigates bounds on the estimation error of a linear system affected by norm-bounded disturbances and full sensor attacks. 
The system is equipped with a detector that evaluates the norm of the innovation signal to detect faults, and the attacker wants to avoid detection.  
We utilize induced $L_\infty$ system norms, also called \emph{peak-to-peak} norms, to compare the estimation error bounds under nominal operations and under attack. This leads to a sufficient condition for when the bound on the estimation error is smaller during an attack than during nominal operation. 
This condition is independent of the attack strategy and depends only on the attacker's desire to remain undetected and (indirectly) the observer gain.
Therefore, we investigate both an observer design method, that seeks to reduce the error bound under attack while keeping the nominal error bound low, and detector threshold tuning.
As a numerical illustration, we show how a sensor attack can deactivate a robust safety filter based on control barrier functions if the attacked error bound is larger than the nominal one. 
We also statistically evaluate our observer design method and the effect of the detector threshold.

\end{abstract}

%%%%%%%%%%%%%%%%%%%%%%%%%%%%%%%%%%%%%%%%%%%%%%%%%%%%%%%%%%%%%%%%%%%%%%%%%%%%%%%%
\input{introduction}
\input{preliminaries}
\input{problemformulation}
\input{errorboundanalysis}
\input{defensemechanisms}
\input{applicationToRobustCBFs}
\input{conclusions}

% \addtolength{\textheight}{-9cm}   % This command serves to balance the column lengths

%%%%%%%%%%%%%%%%%%%%%%%%%%%%%%%%%%%%%%%%%%%%%%%%%%%%%%%%%%%%%%%%%%%%%%%%%%%%%%%%

%%%%%%%%%%%%%%%%%%%%%%%%%%%%%%%%%%%%%%%%%%%%%%%%%%%%%%%%%%%%%%%%%%%%%%%%%%%%%%%%


% The preferred spelling of the word acknowledgment in America is without an e after the g. Avoid the stilted expression, One of us (R. B. G.) thanks . . .  Instead, try R. B. G. thanks. Put sponsor acknowledgments in the unnumbered footnote on the first page.

%%%%%%%%%%%%%%%%%%%%%%%%%%%%%%%%%%%%%%%%%%%%%%%%%%%%%%%%%%%%%%%%%%%%%%%%%%%%%%%%

% References are important to the reader; therefore, each citation must be complete and correct. If at all possible, references should be commonly available publications.

\bibliography{ref} %IEEEabrv

\end{document}

%% file: introduction.tex
\section{Introduction}
\label{sec:introduction}

Today's control systems have evolved from simple plants with local controllers to complex networked and \acrfull{cps}, where the control loops are closed over communication networks and the plants interact with the physical environment \cite{CPSPaper}. 
Examples include large infrastructure systems, such as electrical power grids, and mobile robots relying on remote controllers or estimators. 
For these systems, safety is critical, since unsafe behavior could lead to power outages and costly downtime, but also injuries to humans, if robots operate in shared environments.

Safety is often defined as constraints on the states, and is achieved by designing the input of the system such that the constraints are satisfied. Different control methods for enforcing safety constraints exist, such as model-predictive control~\cite{MPCBook} and \acrlong{cbf}s (\acrshort{cbf}s)~\cite{ames_control_2019}. 
\acrshort{cbf}s are particularly relevant for networked cyber-physical systems, since they can be used locally, as so-called safety filters. This ensures that the applied input is always safe, even if other components of the control system are distributed over the network. 
\emph{Robust} versions of these methods, such as tube MPC (see Chapter~4 of ~\cite{MPCBook}) and robust \acrshort{cbf}s (see, for example, \cite{choi_robust_2021, agrawal_safe_2023}), can be used for systems with uncertainties. The general idea in these methods is to use a bound on the estimation error, to guarantee safety in the presence of bounded disturbances acting on the system and the measurements.

In addition to safety, an important aspect with regard to modern \acrshort{cps}s is cybersecurity \cite{CPSSecurity}. Due to their connection to communication networks, these systems can be vulnerable to attacks -- such as having their transmitted signals modified, to render the system unstable or unsafe.
Safety filters have been shown to keep the system safe from these kinds of attacks as well, for example in \cite{SafetyFilterActuatorAndSensorAttacks} for actuator and sensor attacks and in \cite{ModularSafetyFilter}, where the safety filter is placed on the plant side and has full state information. 
However, when the safety filter uses a state estimator, a sensor attack might deactivate the safety filter \cite{ArnstroemStealthyDeactivation}, without triggering an anomaly detector.
Other works have looked into attack strategies that maximize the estimation error, such as \cite{AttackOnRemoteStateEsimation}, where the attack aims to maximize the mean square error.

In this paper, we analyze how the state estimation error bounds are affected for a system under sensor attack, in order to reason about how such attacks can be mitigated using robust control methods that rely on error bounds. We introduce the concept of a system being \emph{attack-robust} (Definition \ref{def:AttackRobustness}), which allows robust control methods to directly mitigate attacks. Previous works have looked into similar problems, such as~\cite{ResilientSetBasedEstimation}, where a bounded set for the true state is computed, or~\cite{SeverSensorAttackCBF}, where a \acrshort{cbf} is modified to provide safety guarantees for systems under sensor attacks.
However, both assume that only a subset of the sensors are under attack.

The output-to-output gain introduced in \cite{Output2OutputNorm} shows similarities to our approach but considers only finite energy signals.
Furthermore, our analysis uses induced $L_\infty$ system norms, which, to the best of our knowledge, have not been used in secure control settings. This enables us to consider more general setups, such as persistent disturbances and attacks, as well as a general class of detectors.

Our contributions are twofold. First, we provide a sufficient condition for when the bound on the estimation error is smaller during the attack than during nominal conditions (Theorem~\ref{thm:largerNominalBound}). 
This condition does not rely on the actual attack strategy, instead requiring only the assumption that the attacker wants to remain undetected. 
Hence, it provides an offline method to check if the nominal error bound is sufficient to guarantee the safety of the system under sensor attacks. 
Second, we propose an observer-design method (Optimization Problem \eqref{eq:FloatingDiffOptimisation}) that limits the impact of sensor attacks while enabling a trade-off between the minimization of error bounds in the nominal case and the attacked case. 

Our theoretical results are evaluated on an illustrative example using a robust CBF from \cite{agrawal_safe_2023} and the attack proposed in \cite{ArnstroemStealthyDeactivation}. Then, we statistically evaluate the observer design method and show its effect on the system. Lastly, we consider the influence of the detector threshold on the attack bound.

%% file: preliminaries.tex
\section{System norms}
\label{sec:preliminaries}
In this section, we introduce several concepts related to signal and system norms. These will be used later to define the maximum instantaneous $q$-norm of the estimation error, which can be used to guarantee safety by compensating for this maximum error \cite{OR_CBF}. 
The maximum instantaneous $q$-norm of a signal $u(t)$ is described by an $L_\infty$ norm.
\begin{definition}[$L_\infty$ norm]\label{def:SignalNorm}
    Let $u : \mathbb{R} \to \mathbb{R}^{n_u}$ and $q\geq 1$. Then, the $L_\infty$ norm of $u$ with respect to the vector $q$-norm is defined as
    \begin{equation*}
        \norm{u}_{\infty, q} = \sup_t \abs{u(t)}_q=\sup_t \left(\sum_{i=1}^{n_u}|u_i(t)|^q\right)^{\frac{1}{q}}.
    \end{equation*}
\end{definition}
Here, $\norm{u}_{\infty, q}$ defines the size of the $q$-norm ball in which the vector norm $\abs{u(t)}_q$ lies for all $t$.

Let $g_{yu}: \mathbb{R}\rightarrow \mathbb{R}^{n_y\times n_u}$ be an impulse response of a system with input signal $u(t)$ and output signal $y(t)$ such that $y(t)=g_{yu}(t)*u(t)$, where $*$ denotes convolution. To describe the maximum instantaneous $q$-norm of the output $y(t)$, the system norm of $g_{yu}(t)$ is a useful tool, as it describes the amplification from $u(t)$ to $y(t)$.
The system norm induced by the signal norm in Definition~\ref{def:SignalNorm} is defined as follows.
\begin{definition}\label{def:SystemNorm}
    Let $g_{yu}(t)$ be the impulse response of some system. Then, the system norm induced by the corresponding $L_\infty$ norm is given by
    \begin{equation*}
        \norm{g_{yu}}_{1,q} = \sup_{\norm{u}_{\infty, q} = 1} \norm{g_{yu} * u}_{\infty, q}.
    \end{equation*}
\end{definition}
Note that the induced system norm is not necessarily the signal norm of the impulse response. 
The system norms in Definition~\ref{def:SystemNorm} are also known as peak-to-peak norms and if $q=\infty$ we recover the $\ell_1$ system gain \cite{vidyasagar1986optimal}.
These norms are useful for systems influenced by persistent but bounded signals, such as attacks.
Some properties of induced norms are described below, which can be found in \cite{DesoerVidyasagar}.
\begin{lemma}
    \label{lem:InducedNormProperties}
    Let $f$, $g$, and $h$ be impulse responses of appropriate dimensions, $\norm{f}_{1,q}<\infty$, $\norm{g}_{1,q}<\infty$, ${\norm{h}_{1,q}<\infty}$, and $x(t)$ be a signal of appropriate dimension with ${\abs{x(t)}_q<\infty}$, where $q\geq 1$.
    Then the induced norm has the following properties:
    \begin{enumerate}
        \item $\norm{f*x}_{\infty,q} \leq \norm{f}_{1,q} \norm{x}_{\infty,q}$.
        \item $\norm{f + g}_{1,q} \leq \norm{f}_{1,q} + \norm{g}_{1,q}$,
        \item $\norm{f \circ h}_{1,q} \leq \norm{f}_{1,q} \cdot \norm{h}_{1,q}$.
    \end{enumerate}
\end{lemma}

%% file: problemformulation.tex
\section{Problem Formulation}
In this section, we introduce the system and attacker model as well as formulate the problem we want to address.

We consider a continuous-time linear time-invariant plant with the following dynamics;
\begin{equation}
    \label{eq:plant}
    \left\{
    \begin{aligned}
        \Dot{x}(t) &= Ax(t) + Bu(t) + N_1d(t)\\
        y(t) &= Cx(t) + N_2d(t)
    \end{aligned}\right.
\end{equation}
where $x(t)\in\mathbb{R}^{n_x}$, $u(t)\in\mathbb{R}^{n_u}$, and $y(t)\in\mathbb{R}^{n_y}$ are the states, inputs, and outputs of the plant, respectively. 
The signal $d(t)\in\mathbb{R}^{n_d}$ is a disturbance, and note that if $N_1^\top N_2=0$, it has elements that independently influence the state and the measurements, representing process and measurement noise.
\begin{assumption}
    \label{assum:boundeddisturbance}
    The disturbance is bounded such that $|d(t)|_q\leq \epsilon_d$ for all $t\geq 0$, i.e., $\norm{d}_{\infty,q}\leq \epsilon_d$
\end{assumption}

Since we only have access to the plant's outputs, an observer $\hat{x}(t)\in\mathbb{R}^{n_x}$ is used to estimate $x(t)$,
\begin{equation}
    \label{eq:observerdynamics}
    \Dot{\hat{x}}(t) = A\hat{x}(t) + Bu(t) + Kr(t),
\end{equation}
where $K$ is the observer gain and the innovation $r(t)\in\mathbb{R}^{n_y}$ is the difference between the measured and estimated output,
\begin{align}
    \label{eq:innovation}
    r(t)=y(t)-C\hat{x}(t).
\end{align}

The norm of the innovation is often used to determine faults in the system, which leads to the following detector
\begin{align}
    \label{eq:detector}
    |r(t)|_q\begin{cases}
        \leq \nu\quad \text{no alarm},\\
        > \nu\quad \text{alarm},
    \end{cases}
\end{align}
where $q\geq 1$ and $\nu$ is the user-defined detector threshold. The choice of $\nu$ is important, since the goals of detecting all faults and attacks while avoiding false alarms are often in conflict.

Given \eqref{eq:plant} and \eqref{eq:observerdynamics}, we write the dynamics of the estimation error $e(t) = x(t) - \hat{x}(t)$ as follows,
\begin{align}
    \label{eq:errordynamics}
    \Dot{e}(t)=(A-KC)e(t)+(N_1-KN_2)d(t).
\end{align}
\begin{assumption}
    \label{assum:ISS-observer}
    The observer gain $K$ is designed such that the error is input-to-state stable.
\end{assumption}
With Assumption~\ref{assum:ISS-observer}, the estimation error is bounded, i.e., $\abs{e(t)}_q \leq \epsilon_e$ for all $t$, once the influence of the initial states has decayed.

With the nominal system defined, we turn our attention to the adversarial model. 
Having an adversary in the loop will change the dynamics of the estimation error as well as the error bound. As discussed in Section~\ref{sec:introduction}, an accurate bound on the estimation error is important to be able to use robust methods (e.g., \cite{MPCBook,choi_robust_2021, agrawal_safe_2023}) to guarantee constraint satisfaction and safety.
Therefore, we introduce the concept of \emph{attack-robustness}, which relates the nominal error $e(t)$ and the error under attack, denoted $\tilde{e}(t)$.
\begin{definition}
\label{def:AttackRobustness}
A system of the form \eqref{eq:plant}
with observer \eqref{eq:observerdynamics} is \emph{attack-robust} if the nominal bound on the estimation error is larger than the bound under attack, i.e., ${\norm{\tilde{e}}_{\infty,q} \leq \norm{e}_{\infty,q}}$.
\end{definition}

In this paper, we want to investigate the attack-robustness to an attack that targets the sensors. We consider an adversary that replaces the measurement $y(t)$ with $\tilde{y}(t)$, which it can choose freely, as long as it satisfies the constraint used in the detector \eqref{eq:detector}. 
By choosing
\begin{align}
    \label{eq:sensorattack}
    \tilde{y}(t)=C\hat{x}(t)+a(t),\quad \abs{a(t)}_q \leq \nu
\end{align}
the attacker has full control over the innovation, i.e., ${r(t) = a(t)}$.
The constraint $\abs{a(t)}_q \leq \nu$ ensures that the attack does not trigger the detector. 
The attacker's knowledge of $\hat{x}(t)$ in \eqref{eq:sensorattack} may seem like a strong assumption, but in \cite{ConfidentialityAttack} it was shown that an attacker with model knowledge and access to $y(t)$ can obtain $\hat{x}(t)$.

Through replacing $r(t)$ with $a(t)$ in the observer dynamics \eqref{eq:observerdynamics}, the adversary changes the error dynamics to
\begin{align}
    \label{eq:attackederrordynamics}
    \Dot{\tilde{e}}(t)=A\tilde{e}(t)+N_1d(t)-Ka(t).
\end{align}
Clearly, the nominal error dynamics \eqref{eq:errordynamics} are different from the attacked error dynamics \eqref{eq:attackederrordynamics}.
Thus, given Definition~\ref{def:AttackRobustness} and the attack \eqref{eq:sensorattack} we want to tackle the following problem.
\begin{problem}
    \label{prob:AttackedErrorBound}
    Under a sensor attack of the form \eqref{eq:sensorattack},
    when is system \eqref{eq:plant} with observer \eqref{eq:observerdynamics} \emph{attack-robust}?
\end{problem}
The conditions for attack-robustness derived in Section~\ref{sec:ErrorBoundAnalysis} depend implicitly on the observer gain $K$, since $K$ influences the attacked error dynamics~\eqref{eq:attackederrordynamics}. 
Normally, $K$ is designed to reduce the nominal estimation error~\eqref{eq:errordynamics} as much as possible, but with the attacked case in mind as well we can formulate a second problem concerning observer design.
\begin{problem}
    \label{prob:OberserverDesign}
    How can the observer gain $K$ be designed to improve attack-robustness, without sacrificing nominal performance?
\end{problem}

%% file: errorboundanalysis.tex
\section{Analysis of the attacked error bound} 
\label{sec:ErrorBoundAnalysis}

To address Problem~\ref{prob:AttackedErrorBound}, we first derive the nominal error bound $\epsilon_e$ for \eqref{eq:errordynamics}. We then derive the attacked bound $\epsilon_{\tilde{e}}$ for \eqref{eq:attackederrordynamics}, such that $\norm{\tilde{e}(t)}_{\infty,q} \leq \epsilon_{\tilde{e}}$. 
These bounds are then used to provide a condition for attack-robustness of system \eqref{eq:plant}.

From \eqref{eq:errordynamics}, the error $e(t)$ is given by
\begin{equation} 
    e(t)= e^{(A - KC)t}(N_1 - KN_2)*d(t) = g_{ed}(t)*d(t),
\end{equation}
where $g_{ed}(t)$ is the impulse response of the error dynamics. It is therefore bounded according to
\begin{align}
    \label{eq:nominalErrorbBund}
    \norm{e}_{\infty,q}=\norm{g_{ed}*d}_{\infty,q}
    \leq \norm{g_{e d}}_{1,q}\epsilon_d = \epsilon_{e},
\end{align}
where $\norm{g_{e d}}_{1,q}$ exists due to Assumption~\ref{assum:ISS-observer} and thus, $\epsilon_e < \infty$. 

Similarly, the attacked error dynamics \eqref{eq:attackederrordynamics} yield,
\begin{align}
    \tilde{e}(t)&= e^{At}*(N_1d(t) - Ka(t)) \label{eq:matrixConvolutionAttackedError}\\
    &= g_{\tilde{e}d}(t)*d(t) + g_{\tilde{e}a}(t)*a(t)
    \label{eq:convolutionAttackedError}
\end{align}
and the bound on the attacked error bound is given by
\begin{equation}
\label{eq:attackedErrorBound}
    \begin{aligned}
        \norm{\tilde{e}}_{\infty,q}&=\norm{g_{\tilde{e}d}*d+g_{\tilde{e}a}*a}_{\infty,q}\\
        &\leq \norm{g_{\tilde{e}d}*d}_{\infty,q}+\norm{g_{\tilde{e}a}*a}_{\infty,q}\\
        &\leq \norm{g_{\tilde{e}d}}_{1,q}\epsilon_d+\norm{g_{\tilde{e}a}}_{1,q}\nu = \epsilon_{\tilde{e}}.
    \end{aligned}
 \end{equation}
Further, note the use of the adversary's desire to remain undetected, i.e., $ \norm{a}_{\infty,q} \leq \nu$.

Based on the attacked error bound $\norm{\tilde{e}}_{\infty,q}$, we now provide a necessary condition for attack robustness.
\begin{proposition}
    \label{prop:StableSystemNecessary}
    The LTI system \eqref{eq:plant} and observer \eqref{eq:observerdynamics} under the sensor attack \eqref{eq:sensorattack} are \emph{attack-robust} according to Definition~\ref{def:AttackRobustness} only if $A$ is Hurwitz.
\end{proposition}
\begin{proof}
    From \eqref{eq:matrixConvolutionAttackedError}, we observe that if $A$ is not Hurwitz the attacker is able to choose $a(t)$ such that $\tilde{e}(t)$ grows unbounded by exciting the unstable modes of $A$. Thus, $\norm{\tilde{e}}_{\infty,q}$ remains bounded only if $A$ is Hurwitz.
\end{proof}
\begin{remark}
    Proposition~\ref{prop:StableSystemNecessary} shows us that an unstable system under a sensor attack \eqref{eq:sensorattack} cannot be attack robust according to Definition~\ref{def:AttackRobustness}. Hence, to achieve attack robustness for an unstable system, a local controller that stabilizes the unstable system locally is necessary. For example, using $u(t)=-Lx(t)+u_{\mathrm{nom}}(t)$, where the local controller gain $L$ is designed such that $A-BL$ is stable.
\end{remark}
Next, we compare the bounds $\epsilon_e$ and $\epsilon_{\tilde{e}}$, which yields a sufficient condition for attack-robustness.
\begin{theorem}
    \label{thm:largerNominalBound}
    Given Assumptions~\ref{assum:boundeddisturbance} and~\ref{assum:ISS-observer}, consider the LTI system \eqref{eq:plant} and observer \eqref{eq:observerdynamics} under the sensor attack \eqref{eq:sensorattack}.  
    The system is \emph{attack-robust} according to Definition~\ref{def:AttackRobustness} if
    \begin{equation}
        \label{eq:conditionOnBounds}
         \norm{g_{\tilde{e}d}}_{1,q}\epsilon_d+\norm{g_{\tilde{e}a}}_{1,q} \nu \leq \norm{g_{ed}}_{1,q} \epsilon_d.
    \end{equation}
\end{theorem}
\begin{proof}
    For attack-robustness, ${\norm{\tilde{e}}_{\infty,q} \leq \norm{e}_{\infty,q}}$ must hold. 
    Using the error bounds $\epsilon_e$ and $\epsilon_{\tilde{e}}$ given by \eqref{eq:nominalErrorbBund} and \eqref{eq:attackedErrorBound}, respectively, attack-robustness is guaranteed if \eqref{eq:conditionOnBounds} holds.
\end{proof}
Note that \eqref{eq:conditionOnBounds} does not depend on the actual attack strategy, but only on the system dynamics, the observer, and the detector threshold. Hence, it can be evaluated offline.

From Theorem~\ref{thm:largerNominalBound}, observe that the choice of detector threshold is crucial to guarantee $\epsilon_{\tilde{e}}\leq \epsilon_e$.
The choice that leads to no false alarms is particularly interesting and is considered next.

Rewriting \eqref{eq:innovation} as $r(t) = Ce(t) + N_2d(t)$ yields
\begin{align}
    \label{eq:impulseResponseInnovation}
    r(t)=(Cg_{ed}(t)+N_2\delta(t))*d(t) = g_{rd}(t)*d(t),
\end{align}
where $\delta(t)$ is the Dirac delta function.
Due to Assumptions~\ref{assum:boundeddisturbance} and~\ref{assum:ISS-observer}, the innovation can be bounded as follows
\begin{align}
    \norm{r}_{\infty,q} = \norm{g_{rd}*d}_{\infty,q} \leq \norm{g_{rd}}_{1,q}\epsilon_d.
\end{align}
Thus, by choosing $\nu = \norm{g_{rd}}_{1,q}\epsilon_d$, no alarms will be triggered under nominal conditions.

\begin{corollary} 
    \label{cor:noFalseAlarmsThreshold}
    Assume that $\nu = \norm{g_{rd}}_{1,q}\epsilon_d$. Then, the system is \emph{attack-robust} according to Definition~\ref{def:AttackRobustness} if    \begin{equation}
        \label{eq:noFalseAlarmThresholdCondition}
        \norm{g_{\tilde{e}d}}_{1,q} + \norm{g_{\tilde{e}a}}_{1,q} \norm{g_{r d}}_{1,q} =  \norm{g_{e d}}_{1,q}.
    \end{equation}
\end{corollary}
\begin{proof}
    Substituting the threshold $\nu=\norm{g_{rd}}_{1,q}\epsilon_d$ into condition \eqref{eq:conditionOnBounds} leads to
    \begin{equation}
        \label{eq:conditionOnBoundsNoFalseAlarms}
        \norm{g_{\tilde{e}d}}_{1,q}+\norm{g_{\tilde{e}a}}_{1,q}\norm{g_{rd}}_{1,q} \leq \norm{g_{ed}}_{1,q}.
    \end{equation}  
    Writing the error as $e(t)=e^{At}*(N_1d(t) - Kr(t))$, we express the impulse response, $g_{ed}(t)$, as
    \begin{equation}
        g_{ed}(t) = g_{\tilde{e}d}(t) + g_{er}(t) * g_{rd}(t),
    \end{equation}
    where we used $g_{\tilde{e}d}(t)$ from \eqref{eq:convolutionAttackedError}, $g_{er}(t)=-e^{At}K$, and $g_{rd}(t)$ from \eqref{eq:impulseResponseInnovation}.
    Since the attack replaces the innovation, which interrupts the connection described by $g_{rd}(t)$ in the nominal case, it does not affect the connection from innovation to error. 
    Therefore, $g_{er}(t) = g_{\tilde{e}a}(t)$ and Lemma~\ref{lem:InducedNormProperties} yields
    \begin{equation}
        \label{eq:NomErrorUpperBoundNoFalseAlarms}
        \norm{g_{ed}}_{1,q} \leq \norm{g_{\tilde{e}d}}_{1,q}+\norm{g_{\tilde{e}a}}_{1,q}\norm{g_{rd}}_{1,q}.
    \end{equation}
    Comparing \eqref{eq:conditionOnBoundsNoFalseAlarms} and \eqref{eq:NomErrorUpperBoundNoFalseAlarms} yields the condition in \eqref{eq:noFalseAlarmThresholdCondition}.
\end{proof}

Corollary~\ref{cor:noFalseAlarmsThreshold} shows that the nominal error bound, $\epsilon_e$, will never be larger than and at most be equal to the attacked error bound, $\epsilon_{\tilde{e}}$, if the detector threshold is chosen such that no false alarms will occur under nominal conditions. 
However, because of the inequalities in the derivation, condition \eqref{eq:noFalseAlarmThresholdCondition} rarely holds in practice.
Albeit a sufficient condition, this suggests that to guarantee attack-robustness, one needs to allow false alarms in the nominal case.
\begin{remark}
    For a system that fulfills \eqref{eq:noFalseAlarmThresholdCondition}, the disturbance that produces the maximum error must also produce the maximum innovation.
    Otherwise, the attacker can choose an attack that points in the direction of the maximum error with the size of the maximum innovation. In the nominal case, this is an unnatural occurrence, but due to the independence of the attack from the disturbance it is possible in the attacked case. Therefore, the attacked error could exceed $\epsilon_e$ if \eqref{eq:noFalseAlarmThresholdCondition} does not hold.
    Note that since the adversary has no control over the term $g_{\tilde{e}d}(t)*d(t)$ in \eqref{eq:convolutionAttackedError}, an unfavorable realization of the disturbance $d$ may also be required for $\tilde{e}$ to exceed $\epsilon_{e}$.
\end{remark}

%% file: defensemechanisms.tex
\section{Attack-robust observer design}
\label{sec:defenseMechanism}
In the previous section, we have provided a sufficient condition for attack-robustness and, therefore, also for when an adversary is able to increase the estimation error above the nominal bound.
This ability presents a risk, since any safety filter that relies on a bound on the error would be compromised. 
One solution would be to use $\epsilon_{\tilde{e}}$ instead of $\epsilon_e$ in the robust controller, which would guarantee safety even in the presence of an attack. 
However, this might deteriorate the control performance of the closed-loop system if $\epsilon_{\tilde{e}}\gg \epsilon_e$.
Another solution would be to change the detector threshold $\nu$ to reduce the impact of an attacker who wants to remain undetected, at the cost of false alarms. In Section~\ref{sec:DetectorThresholdChoice}, the effects of the choice of the detector threshold are investigated.

In this section, we instead address Problem~\ref{prob:OberserverDesign} by presenting a third alternative: designing the observer gain $K$ to limit the impact of an attack. 
The observer gain $K$ influences the dynamics of both the nominal error \eqref{eq:errordynamics} and the attacked error \eqref{eq:attackederrordynamics}.
The trivial solution ${K = 0}$ eliminates the attacker's influence over the system, but also eliminates the observer's ability to use measurements. Ideally, we want to design the gain $K$, such that attack-robustness is achieved or, at least, such that $\epsilon_e\approx \epsilon_{\tilde{e}}$ holds.
At the same time, we do not want to increase the nominal error bound $\epsilon_e$ significantly.
Thus, we are faced with a multi-objective optimization problem, where we would like to minimize both $\epsilon_e$ and $\epsilon_{\tilde{e}}$.

Although there are methods to design observers to reduce the induced $L_\infty$ gain, these methods are usually complex to implement and can be computationally expensive \cite{Peak2PeakIsDifficultToWorkWith}.
Therefore, we make use of the fact that some peak-to-peak norms are equivalent to the $\mathcal{H}_{\infty}$ norm, for example, \cite{skogestad2005multivariable}
\begin{align}
\norm{e}_{\mathcal{H}_{\infty}} \leq \sqrt{n_y}\norm{e}_{1,\infty} \leq \sqrt{n_u n_y}(2n_x+1)\norm{e}_{\mathcal{H}_{\infty}},
\end{align}
and will use $\mathcal{H}_{\infty}$ optimization to design $K$, where $\norm{e}_{\mathcal{H}_\infty}$ denotes the $\mathcal{H}_{\infty}$ norm of $e$.
Recall that $\gamma$ is an upper bound on the $\mathcal{H}_{\infty}$ norm of \eqref{eq:errordynamics} if there exists a positive definite matrix $P$ such that the matrix $ \Theta\left(P, A-KC, N_1-KN_2, \gamma\right)$ is negative definite, where
\begin{align}
\label{eq:boundedRealLemmaHelperFunction}
\Theta(P,A,B,\gamma) = \begin{pmatrix}
            PA + A^TP & PB & I \\
            B^TP & -\gamma I & 0 \\
            I & 0 & -\gamma I
    \end{pmatrix},
\end{align}
with the identity and zero matrices having appropriate dimensions \cite{BoundedRealLemma}.

Hence, instead of finding a gain $K$ that minimizes $\epsilon_e$ and $\epsilon_{\tilde{e}}$ directly, we focus on minimizing the  $\mathcal{H}_{\infty}$ norms, $\gamma$ and $\gamma_{\tilde{e}}$ of the error systems \eqref{eq:errordynamics} and \eqref{eq:attackederrordynamics}, respectively.
To find a solution to this multi-objective problem we use a scalarization approach, which leads to the following optimization problem:
\begin{equation}
\label{eq:FloatingDiffOptimisation}
\begin{aligned}
        \min_{K, P, \gamma_{\tilde{e}}, \gamma} \beta\gamma_{\tilde{e}} + \gamma\ \text{s.t.}\ \begin{cases}
            
            P \succ 0,\\
            \Theta\left(P, A-KC, N_1-KN_2, \gamma\right) \prec 0,\\
            \Theta\left(P,A,\begin{pmatrix} N_1, & -K\end{pmatrix},\gamma_{\tilde{e}}\right) \prec 0,
        \end{cases}
\end{aligned}
\end{equation}
where $\beta>0$ is a user-defined weight.
Here, the constraints ensure that $\gamma$ and $\gamma_{\tilde{e}}$ are upper bounds for the $\mathcal{H}_{\infty}$ norm of the respective systems. Note that the last constraint, derived from \eqref{eq:matrixConvolutionAttackedError}, considers a stacked input vector $\begin{pmatrix} d^\top & a^\top\end{pmatrix}^\top$.
Although not convex, \eqref{eq:FloatingDiffOptimisation} can be re-formulated as an equivalent convex optimization problem using standard techniques \cite{skogestad2005multivariable}.

Note that due to the first two constraints in \eqref{eq:FloatingDiffOptimisation} the designed observer will be stable. 
Since an equivalent convex problem to \eqref{eq:FloatingDiffOptimisation} exists, one can achieve all values on the Pareto front of the multi-objective optimization problem that minimize $\gamma$ and $\gamma_{\tilde{e}}$ by varying $\beta$ \cite{boyd2004convex}. Hence, when designing $K$ one can choose the trade-off between minimizing $\gamma$ and $\gamma_{\tilde{e}}$ with $\beta$. Moreover, with $K=0$ and $\gamma=\gamma_{\tilde{e}}$ one can always find a feasible solution to \eqref{eq:FloatingDiffOptimisation} since $A$ is stable.
Finally, \eqref{eq:FloatingDiffOptimisation} aims to reduce $\norm{\tilde{e}}_{\infty,q}$ regardless of the attacker wanting to remain undetected or not, i.e., the magnitude of the attack. This also enables us to independently design the observer gain and tune the detector threshold to further limit the attack impact if necessary.

%% file: applicationToRobustCBFs.tex
\section{Numerical Experiments} 
\label{sec:NumericalExperimentsCBFs}
In this section, we numerically evaluate our results from the previous sections. 
We begin with an illustrative example that shows that safe robust controllers do not guarantee safety in the presence of a sensor attack.
Then we will evaluate the observer design by comparing it to the Kalman filter and consider randomly generated strictly proper and stable systems.
Finally, we investigate how the choice of detector threshold influences the attack impact and false alarm rate using again randomly generated systems. 

The infinity vector norm, i.e., $q=\infty$, is used throughout to obtain the signal and the induced system norms.
For all systems, we use $N_1 = \begin{pmatrix}
    I & 0
\end{pmatrix}$ and ${N_2 = \begin{pmatrix}
    0& I
\end{pmatrix}}$ such that $d(t) \in \R^{n_x+n_y}$.
The disturbance $d(t)$ is bounded by ${\norm{d}_{\infty,\infty}\leq \epsilon_d = 0.5}$, where each element is uniformly drawn from the interval $[-0.5, 0.5]$ and independent of the other elements in $d(t)$. 

\subsection{Illustrative examples}\label{sec:IllustrativeExamples}
For the illustrative example, we use the following system
\begin{equation}\label{eq:SimulationSystemFull}
    \left\{\begin{aligned}
        \Dot{x}(t) &= \begin{pmatrix} 
        -0.2 & 1 \\ 
        0 & -1 
        \end{pmatrix} x(t)
        + \begin{pmatrix}
            0 \\ 1
        \end{pmatrix}u(t)
        + \begin{pmatrix}
            I & 0
        \end{pmatrix}d(t) \\
        y(t) &= x(t) + \begin{pmatrix}
            0 & I
        \end{pmatrix}d(t)
    \end{aligned}\right.
\end{equation}
The observer is a Kalman filter, designed with the noise covariance of independent and identically distributed uniform process and measurement noise, both with the bound $0.5$.
System~\eqref{eq:SimulationSystemFull} has the norms ${\norm{g_{ed}}_{1,\infty} \approx \num{2.44}}$, ${\norm{g_{\tilde{e}a}}_{1,\infty} \approx \num{8.36}}$, ${\norm{g_{rd}}_{1,\infty} \approx \num{3.44}}$, and ${\norm{g_{\tilde{e}d}}_{1,\infty} \approx \num{9.97}}$, which leads to
\begin{align*}
    \epsilon_{e} &\approx \num{1.22}, & \epsilon_{\tilde{e}} &\approx \num{19.37}, & \nu &\approx \num{1.72}
\end{align*}
with $\nu$ at the level needed to ensure no false alarms. Note that since ${\norm{g_{\tilde{e}d}}_{1,\infty} > \norm{g_{ed}}_{1,\infty}}$, no $\nu > 0$ could fulfill the condition in Theorem \ref{thm:largerNominalBound}.

We consider the system safe if it remains within the set given by $h(x)\geq 0$ with $h(x) = -x_1 - x_2$. To obtain a safe control input, a safety filter based on the observer-robust CBF introduced in \cite{agrawal_safe_2023} is used.
This safety filter is given by a quadratic optimization problem as follows
\begin{equation*}
\begin{aligned}
    &\min_{u}\|u-u_{des}\|^2_2, \\ 
    \mathrm{s.t.}\enskip &  \text{\small $\nabla h(\hat{x})(A\hat{x}+Bu)\geq -\alpha\big(h(\hat{x})-L_h M(t)\big)+L_h\dot{M}(t)$},
\end{aligned}
\end{equation*}
where $\nabla h(\hat{x})$ and $L_h$ are the gradient and Lipschitz constant of $h(x)$, and $e(t) \leq M(t)$ for all $t$. 

In our simulations, we use $x(0)=\hat{x}(0)=\begin{pmatrix}
    -2, & 0
\end{pmatrix}^\top$, such that $M(t)$ is a constant that corresponds to the error bound and we use ${\alpha(h) = h}$. 
We further use $u_{des} = 1$, which would lead to an expected steady state of $x_{ss}=\begin{pmatrix}
    5, &1
\end{pmatrix}^\top$. Note that $h(x_{ss})<0$ such that this state is unsafe and the safety filter should prevent the system to reach that point. For the sensor attack, we use the attack proposed by \cite{ArnstroemStealthyDeactivation}. The goal of this attack is to bias the state estimate towards the inside of the safe set, such that the safety filter will be deactivated.

In Fig.~\ref{fig:Attack} we show the nominal and attacked trajectory of the system together with the trajectory of the state estimates.
In Fig.~\ref{fig:Attack}(a), the nominal error bound $\epsilon_e$ is used in the robust CBF, i.e., $M(t)=\epsilon_e$. The attack is able to steer the state to the unsafe steady state $x_{ss}$ indicating that the safety filter was successfully deactivated, while under nominal conditions, the safety is not violated despite the presence of the disturbance.
In Fig.~\ref{fig:Attack}(b), the attacked error bound is used, i.e., $M(t)=\epsilon_{\tilde{e}}$, and in both cases the trajectories remain safe.
It is important to note that the bounds we obtained are worst-case bounds, i.e., the realization of the disturbance is the worst-case as well. However, the simulations used a random disturbance and the attacker still managed to deactivate the safety filter.

\begin{figure}
    \centering
    \subfloat[Trajectories when using the $\epsilon_{e}$ bound.]{
        \includegraphics[clip,width=\linewidth]{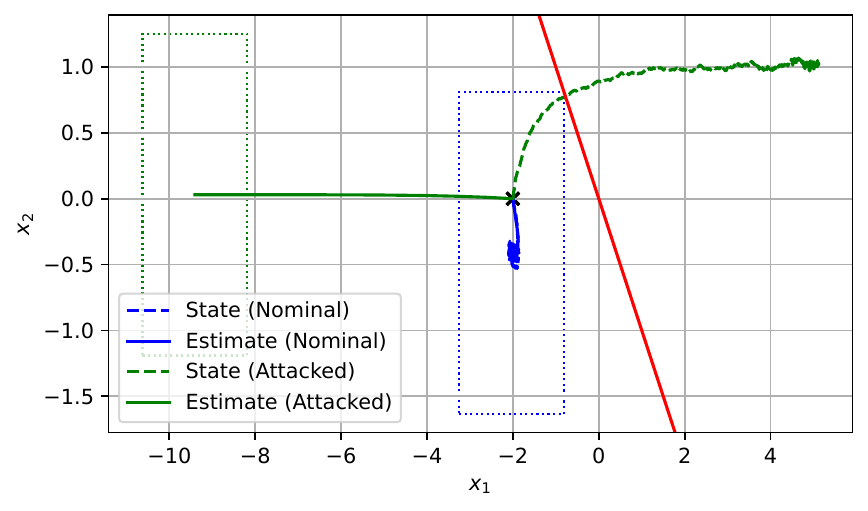}
    }
    \\
    \subfloat[Trajectories when using the $\epsilon_{\tilde{e}}$ bound.]{
        \includegraphics[clip,width=\linewidth]{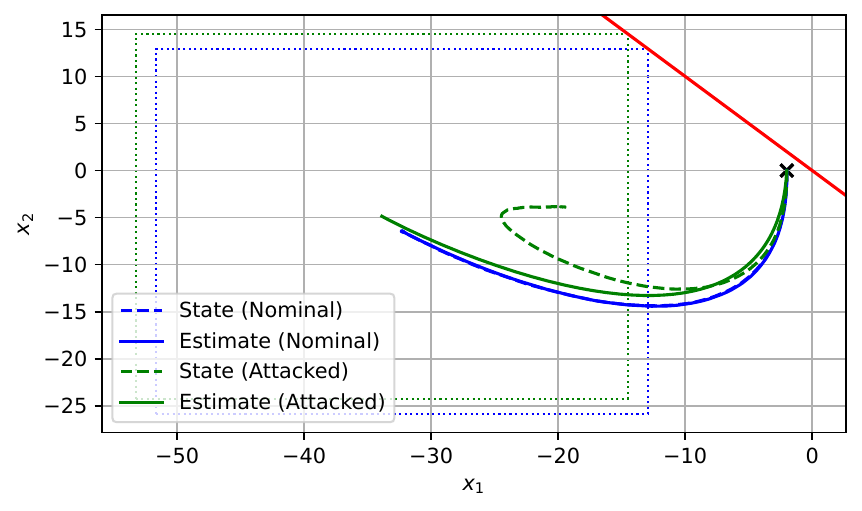}
    }
    \caption{Trajectories for the system under nominal conditions and under attack. The red line shows the boundary of the safe set and the dotted lines show the boundaries of the sets where the robust CBFs assume the true state to be.}
    \label{fig:Attack}
\end{figure}

\subsection{Observer design}
Let us now evaluate our attack-aware observer design proposed in Section~\ref{sec:defenseMechanism}.
To obtain statistically relevant results, we randomly generate 100 strictly proper and stable systems with $n_x=n_u=n_y=2$ and solve problem \eqref{eq:FloatingDiffOptimisation} to obtain the gain $K$ for each of the systems. 
Recall that our observer design enables us, through $\beta$, to choose the trade-off between the error bounds in the nominal and attacked cases.
We set the threshold such that no false alarms will be triggered, as given in Section~\ref{sec:ErrorBoundAnalysis}.

Furthermore, we compare our method to the Kalman filter design as a baseline. Figure~\ref{fig:ErrorBoundComparison} shows that observers designed according to \eqref{eq:FloatingDiffOptimisation} bring $\epsilon_{\tilde{e}}$ closer to $\epsilon_{e}$ compared with a Kalman filter. This comes at the cost of the nominal error bound $\epsilon_e$ typically being larger. Furthermore, we observe that by choosing a larger $\beta$, the designed observer leads to lower values for $\epsilon_{\tilde{e}}$. This showcases the ability to trade off nominal and attacked error bounds through $\beta$.
\begin{table}[ht]
    \centering
    \caption{Quantiles of the ratio between $\epsilon_{\tilde{e}}$ and $\norm{\tilde{e}}_{\infty,q}$}
    \label{tab:Conservatism}
    
    \begin{tabular}{@{}l*{5}{S[table-format=1.2]}@{}}

        \toprule
        Quantile & 0.9 & 0.95 & 0.99 & 1 \\ \midrule
        $\frac{\epsilon_{\tilde{e}}}{\norm{\tilde{e}}_{\infty,q}}$ & 1.07 & 1.13 & 1.27 & 1.44 \\ \bottomrule
    \end{tabular}
    \end{table}
    
Since $\epsilon_{\tilde{e}}$ is an upper bound on the attacked error bound, $\norm{\tilde{e}}_{\infty,q}$, let us investigate how conservative this bound is. We use 1000 random systems and calculate $\norm{\tilde{e}}_{\infty,\infty}$. 
The quantiles of the ratio between $\epsilon_{\tilde{e}}$ and $\norm{\tilde{e}}_{\infty,\infty}$ are shown in Table~\ref{tab:Conservatism}. 
We observe that ${\epsilon_{\tilde{e}}\leq 1.44 \norm{\tilde{e}}_{\infty,\infty}}$ and ${\epsilon_{\tilde{e}}\leq 1.07 \norm{\tilde{e}}_{\infty,\infty}}$ holds for all and for $90\,\%$ of the systems, respectively, that use our observer with $\beta=100$. Thus, in most cases, the bound is not very conservative.
\begin{figure}
    \centering
    \input{optimization_compare_plot}
    \caption{Error bounds for Kalman filters and observers designed according to \eqref{eq:FloatingDiffOptimisation} for 100 randomly generated systems.}
    \label{fig:ErrorBoundComparison}
\end{figure}
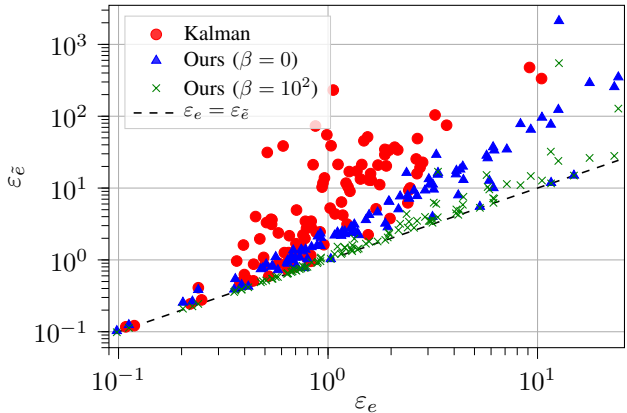

\subsection{Influence of the detector threshold}
\label{sec:DetectorThresholdChoice}
Finally, we will investigate how the tuning of the detector threshold $\nu$ can be used to mitigate the attack impact.
Here, 1000 randomly drawn systems were used, and for each of the systems we designed a Kalman filter and an observer based on our observer design with $\beta=100$.  
Let ${\nu_{\max}= \norm{g_{rd}}_{1,\infty}\epsilon_d}$ be the threshold that will not result in false alarms. 

Consider $\epsilon_{\tilde{e}}$ as a measure of the attack impact. Then, \eqref{eq:attackedErrorBound} shows that $\epsilon_{\tilde{e}}$ decreases linearly as we decrease $\nu$. Furthermore, a necessary condition for finding a $\nu>0$ that fulfills \eqref{eq:conditionOnBounds} is $ \norm{g_{ed}}_{1,q} > \norm{g_{\tilde{e}d}}_{1,q}$.

While choosing $\nu<\nu_{\max}$ does limit the attacker, it also leads to false alarms. 
To investigate their prevalence, we simulated 1000 randomly drawn systems investigated for \qty{100}{\s} under nominal conditions, and calculated the empirical false alarm rate for different values of $\nu\in(0,\nu_{\max}]$.
\begin{figure}
    \centering
    \includegraphics[clip,width=\linewidth]{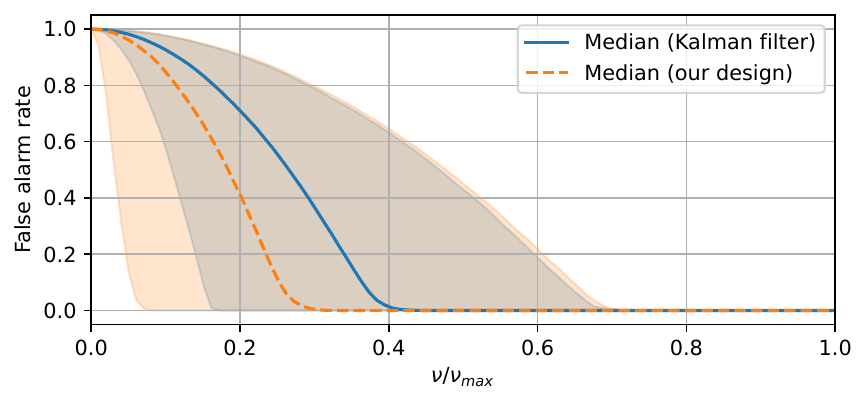}
    \caption{Empirical median false alarm rate based on 1000 randomly generated systems equipped with a Kalman observer or our observer with $\beta=100$. The shaded area indicates the $90\%$ confidence interval of the empirical false alarm rates.}
    \label{fig:EmpFARVaryThreshold}
\end{figure}
The results are shown in Fig.~\ref{fig:EmpFARVaryThreshold}. In this figure, we observe that the median for the false alarm rate when our observer design is used is smaller than the median for the system with the Kalman filter.
Hence, our observer design seems to reduce the false alarm rate compared to the Kalman filter design.
Further, the $90\,\%$ confidence interval depicted as the shaded area also shows that in many cases we can choose $\nu=0.75\nu_{\max}$ without significantly increasing the false alarm rate.

%% file: optimization_compare_plot.tex
\begin{tikzpicture}

\definecolor{darkgray176}{RGB}{176,176,176}
\definecolor{green}{RGB}{0,128,0}
\definecolor{lightgray204}{RGB}{204,204,204}

\begin{axis}[
legend cell align={left},
legend style={
  fill opacity=0.8,
  draw opacity=1,
  text opacity=1,
  nodes={scale=0.8, transform shape},
  at={(0.03,0.97)},
  anchor=north west,
  draw=lightgray204
},
log basis x={10},
log basis y={10},
tick align=outside,
tick pos=left,
x grid style={darkgray176},
xlabel={$\varepsilon_{e}$},
xmin=0.09, xmax=26,
xmode=log,
xtick style={color=black},
y grid style={darkgray176},
ylabel={$\varepsilon_{\tilde{e}}$},
ymin=0.06, ymax=3500,
ymode=log,
ytick style={color=black},
y post scale = 0.8,
grid=major
]
\addplot [draw=red, fill=red, mark=*, only marks]
table{%
x  y
0.521563541723897 0.592750613710957
0.394657430930652 1.61048156140507
0.240391759391619 0.409898568017517
1.95436275493931 33.611212426343
0.512783056222608 31.3180896230226
2.08024317496115 37.0953438599277
0.219940296072707 0.243984087069214
0.610313992016622 38.531553698655
1.11597052808077 21.3410536489839
0.679546961206522 0.787437910729311
0.65341565999473 2.16964244309668
1.44959705913112 21.2764824405912
0.936728277586764 10.3076100236739
0.6364900908741 0.737509808731624
1.97658403243591 3.76697110802548
2.15176328420906 33.9912544162925
1.2325484169924 3.15131092423227
1.86714267809379 29.1125138992254
0.759110298986011 2.22505306045619
0.45098123257154 4.01012234134634
0.249249199108505 0.27735164312632
1.56890603066723 12.9287079614755
0.669950473208994 0.883643250533095
0.828481522708857 0.955014773578475
0.118943215984542 0.121552723851962
0.826975528403372 1.15803580373353
0.722926999311136 1.46371596146189
0.789351990627249 2.7479585480024
0.440495766713718 0.510927006212156
0.70983740029967 1.19680095678869
0.534272459928039 3.68239802764475
9.15938558968094 477.191542527148
1.15913841847295 6.388383324196
1.73161020026751 20.8897930400264
0.702998228782727 1.68193003775562
0.876347993091756 2.43268508776573
1.05903536181013 231.398958825586
0.955331857518295 1.62590888168395
0.53796639964589 3.12483546573011
1.26274077421339 16.9596001561849
0.817966796352209 2.92820456854498
1.27229861172343 9.02144388914537
0.398856918682001 0.624694920677988
0.492813959729392 1.08355724514732
0.497197282553316 0.832154055775895
0.987019238970737 55.1535370205756
1.89742126744305 20.3408321508147
0.108331323583478 0.11702846790459
2.74899511527705 19.3299762425765
2.82793160649735 22.819880802311
1.86754961509016 34.5674610574764
0.775886731542364 1.98544562662422
2.67201067388279 48.8551416009982
2.4037257255686 6.22713951575173
0.829492523166814 2.87526525393538
1.02338830839597 5.25505478403278
3.24262236373458 104.047194107339
1.4281179724004 13.3677079386444
0.747328846966887 1.51854073247404
2.42890539844746 9.32201160182472
1.37385344686553 17.2050290593238
0.580237628225489 0.900023250498943
0.47243055262536 1.96026134534366
0.916228309041075 2.45671019261997
3.68188608597315 74.9550017855903
1.17645600671105 4.19421334728739
0.508398133593779 3.309497312375
0.366522653655478 0.964454541538505
1.24580313615547 10.3769758164946
0.779836858296555 3.45138637137295
0.539914644106698 3.62733111389616
0.848509673434568 21.1416994136933
0.648227639000628 1.26382882790798
0.870900202786726 73.1278464828742
0.44125186350087 0.871720328433142
1.6953364449874 5.13847869678095
1.59599496283107 15.8900384575563
1.48308977697763 45.4766169359928
0.794015477602447 1.06499821226223
0.945840677748477 11.7614757828237
10.4518604990381 334.541389649597
0.379722914869184 0.479249626219962
0.70719871683575 4.94273299271895
2.63560317881115 25.4237381857582
1.60825186811308 12.892495850986
0.640976522597607 0.931181440171092
1.03387979350418 38.856074597272
2.46471103673138 10.0304513362228
1.21894595689139 11.7197375856541
0.961596723394763 13.8460254712849
0.831050417965657 3.45809754146759
0.579095595158488 2.37348943222731
0.684420204706323 1.24351256385621
0.412075459547109 0.528551655742972
1.76449966063939 11.1620119536224
1.07992135197794 4.36582277136108
1.55418091120946 2.23715923469293
0.618696541550737 1.26311463426773
2.65867273152318 15.848403012563
1.54501899604925 51.513198390461
};
\addlegendentry{Kalman}
\addplot [draw=blue, fill=blue, mark=triangle*, only marks]
table{%
x  y
0.488313061805916 0.767536980679558
0.682106647202888 1.36302823641533
0.240136054271206 0.382996212420988
5.71291279469989 28.3866200534281
3.34912777848451 16.3341631555
5.78288016111179 26.1159044992558
0.202941680664502 0.254955801643657
14.9170239815758 14.9853572081938
5.30743268608433 5.51223828416543
0.643820358441474 0.889984698041736
0.755935452341892 1.31938395473762
4.42729339358637 13.032149084308
2.20932254030778 7.69059114744241
0.576357806327891 0.723883502234424
1.92682098130237 2.85748017663666
6.13834741316609 37.2372717585334
1.30328630490338 2.37609567879804
6.16171068149103 10.026892955063
0.919476058932506 2.34417385755916
1.33529321548262 3.12299815027578
0.225254467307262 0.264222161914888
3.16328686716967 11.2662593059489
0.642031984156918 0.989188080720138
0.797721146601016 1.02910114186731
0.111979823270954 0.124308232766763
0.772885792476087 0.773043389828824
0.677303748012003 1.00358661768646
0.925909500303874 2.03862353872478
0.416210428284808 0.426337581757334
0.668480336066022 1.11892514014269
1.30716197152319 2.80270465693929
24.3050630568893 349.53240461347
1.95961435849145 6.48114532530684
3.19550557134234 9.58274989850494
0.780295639280859 1.39614174465353
0.894649958464409 2.13030358991957
12.6552802761848 2122.89422572933
0.888186675254264 1.59394894099733
1.20204996798654 2.29136690807002
2.37028130127383 16.3458393397817
1.16822606976037 2.48846445500366
1.94931701183547 6.54933870915395
0.386868181110485 0.44860973963111
0.562764489650402 1.10019506449049
0.506681710305439 0.870260456625423
12.5727574872146 123.275458836982
4.10043581674292 17.4832659018964
0.09813235077247 0.102626566621317
3.33865163169468 15.4332822947352
4.27603703766597 18.2193379075299
5.88734023981616 12.7054586411233
1.03083481629228 1.03081762709495
7.13613826891343 34.6486832467245
2.16910787835356 5.09026632094928
1.07452554808694 2.19527841182923
1.57102364911261 4.24351565032314
11.5722476364746 76.6951383886991
2.90415486527931 9.31198736767194
0.744294838402072 1.38956661530942
2.66013902583211 7.77878704283354
4.18555051599178 7.96828924709181
0.539363534719785 0.807703963191735
0.834367123454664 1.51341979556884
1.15511860312057 2.17459031205749
9.28628627019621 65.2419023896973
1.292226031318 2.41693365630819
1.20180138348627 2.49358416162988
0.47777055579405 0.746528869070356
2.43394057707905 7.14332113316949
1.28904577910799 2.86071353039012
1.10967286302128 2.41740936331758
6.21180849335098 33.5350355145015
0.706271261730881 0.830709550945239
23.2517373330634 255.721816671693
0.360952086861361 0.543849525703119
1.72472227250974 4.60390320810878
4.07997557839604 15.5957807546411
11.6071461763061 11.8440630757206
0.732384761427206 1.08577428578228
3.06452053034013 11.4638208310039
17.7223099199889 291.348417324686
0.356604150491071 0.388603215261931
1.8770830678055 3.21633515187734
4.39800460587537 20.0083181568941
3.105225777576 11.3022243947539
0.543608606508108 0.821016450251597
8.29020966222521 78.2103123418156
2.57352771589116 9.23428117898781
3.15241535450725 3.9816247707797
3.28619786870732 29.1894564209284
1.38078629738713 2.39505889746686
0.921979948724661 1.56759800188231
0.636754098204518 0.917374281524403
0.386715233720841 0.421758489077808
2.39851943580008 7.87646038126114
1.5161750487573 3.60435577251044
1.40125708017507 2.22334234718817
0.698195686973615 1.08122488043321
3.02753596780414 13.3868595529403
10.499295144004 95.9574863740743
};
\addlegendentry{Ours ($\beta = 0$)}
\addplot [draw=green, fill=green, mark=x, only marks]
table{%
x  y
0.488313061805916 0.504375757268577
0.682106647202888 0.753131178932901
0.240136054271206 0.244132394462547
5.71291279469989 15.1978257337019
3.34912777848451 16.8944956128489
5.78288016111179 10.3481395360934
0.202941680664502 0.209380408962013
14.9170239815758 14.9204856570668
5.30743268608433 5.31113532110468
0.643820358441474 0.653024979669807
0.755935452341892 0.825034459100906
4.42729339358637 7.55411019172561
2.20932254030778 3.1871802334762
0.576357806327891 0.578671090986963
1.92682098130237 2.05535628807155
6.13834741316609 7.21954444988125
1.30328630490338 1.32499137959202
6.16171068149103 6.23233342656331
0.919476058932506 0.960050752590239
1.33529321548262 1.97461082832573
0.225254467307262 0.225753515128368
3.16328686716967 3.6545321186318
0.642031984156918 0.648339350225058
0.797721146601016 0.801138390263566
0.111979823270954 0.112924765170166
0.772885792476087 0.77292240329846
0.677303748012003 0.693837627456722
0.925909500303874 1.07870762151354
0.416210428284808 0.41627163292806
0.668480336066022 0.67747426520699
1.30716197152319 1.59831961625594
24.3050630568893 127.718747318541
1.95961435849145 2.30020113086357
3.19550557134234 6.26593492778589
0.780295639280859 0.797217206748206
0.894649958464409 0.941750165509981
12.6552802761848 548.256283003045
0.888186675254264 0.902656938089246
1.20204996798654 1.49875058070223
2.37028130127383 4.35318271757025
1.16822606976037 1.25570809471973
1.94931701183547 2.39417001953423
0.386868181110485 0.386878828534477
0.562764489650402 0.577297363320522
0.506681710305439 0.517037632250846
12.5727574872146 23.8762198388414
4.10043581674292 5.03832491467512
0.09813235077247 0.0983655152596384
3.33865163169468 3.92352755948702
4.27603703766597 4.80842715869732
5.88734023981616 9.57329756882754
1.03083481629228 1.03084435505692
7.13613826891343 11.3841493232391
2.16910787835356 2.59465111627588
1.07452554808694 1.20727480876605
1.57102364911261 1.80218933662504
11.5722476364746 32.1436115678802
2.90415486527931 3.96157845810146
0.744294838402072 0.76183011132327
2.66013902583211 3.50481524079
4.18555051599178 4.25064602287808
0.539363534719785 0.544537808627632
0.834367123454664 0.925667854560058
1.15511860312057 1.20920755695172
9.28628627019621 14.7141774187779
1.292226031318 1.37734593608208
1.20180138348627 1.47846136513694
0.47777055579405 0.488066217360436
2.43394057707905 3.40671299843832
1.28904577910799 1.43890530500851
1.10967286302128 1.53059736757133
6.21180849335098 7.10511834755197
0.706271261730881 0.706283987103216
23.2517373330634 28.0358397889323
0.360952086861361 0.373849392915949
1.72472227250974 1.80488970940461
4.07997557839604 4.45495692976521
11.6071461763061 11.6075843258075
0.732384761427206 0.743220506206182
3.06452053034013 3.49258048796799
17.7223099199889 26.2813418646544
0.356604150491071 0.356904997949196
1.8770830678055 1.90111760405496
4.39800460587537 5.86267064501582
3.105225777576 3.44896226810979
0.543608606508108 0.547471537917828
8.29020966222521 10.9664887664196
2.57352771589116 2.79993789845709
3.15241535450725 3.29226587294313
3.28619786870732 6.03134293254204
1.38078629738713 1.40030254143316
0.921979948724661 1.04351011807839
0.636754098204518 0.640881213124758
0.386715233720841 0.387099045621999
2.39851943580008 2.81909913273232
1.5161750487573 1.68311823294519
1.40125708017507 1.43943687290494
0.698195686973615 0.725582272107255
3.02753596780414 3.50054530541709
10.499295144004 12.7202317345043
};
\addlegendentry{Ours ($\beta = 10^2$)}
\addplot [semithick, black, dashed]
table {%
0.09813235077247 0.09813235077247
24.3560664607765 24.3560664607765
};
\addlegendentry{$\epsilon_e = \epsilon_{\tilde{e}}$}
\end{axis}

\end{tikzpicture}

%% file: conclusions.tex
\section{Conclusions}
In Theorem~\ref{thm:largerNominalBound}, we have provided a sufficient condition for a system's \emph{attack-robustness}, i.e., when the estimation error in the case of attack is smaller than the bound during nominal operation, independent of the attack strategy.
Any safety guarantee based on a bound on the nominal estimation error would be invalidated, in the case of attack, if this condition were not fulfilled. 
One remedy would be to increase the error bound to the bound under attack. 
However, this change might reduce the control performance, and therefore, we proposed an observer design with the aim of reducing the attack impact. We evaluated the observer design on randomly drawn systems and also investigated the false alarm rate based on the detector threshold.

There are several paths for future research.
For example, an extension to non-linear systems would generalize the results, along with investigating more general sets than norm balls. Furthermore, incorporating actuator attacks in our framework in addition to sensor attacks would increase its applicability.